\documentclass[letterpaper, 10 pt, conference]{ieeeconf}
\IEEEoverridecommandlockouts
\def\BibTeX{{\rm B\kern-.05em{\sc i\kern-.025em b}\kern-.08em
    T\kern-.1667em\lower.7ex\hbox{E}\kern-.125emX}}

\usepackage{graphicx} 
\usepackage[none]{hyphenat}
\usepackage{amsmath}
\usepackage{amssymb}
\usepackage{mathrsfs,bm}
\usepackage{amsthm}
\usepackage{cite}
\usepackage{graphicx}
\usepackage{amsfonts}
\usepackage{amstext}
\usepackage{float}
\usepackage{cleveref}
\usepackage{nomencl}
\newtheorem{theorem}{Theorem}

\newtheorem{definition}{Definition}
\newtheorem{lemma}{Lemma}
\newtheorem{corollary}{Corollary}
\newtheorem{remark}{Remark}
\usepackage{color}

\begin{document}

\title{
System Design Approach for Control of Differentially Private \\Dynamical Systems 
\author{Raman Goyal, Dhrubajit Chowdhury, and Shantanu Rane
\thanks{R. Goyal, D. Chowdhury, and S. Rane are with Palo Alto Research Center - Part of SRI International, Palo Alto, CA, USA. 
{ \tt\small \{raman.goyal, dhruba.chowdhury, shantanu.rane\}@sri.com, }}}
}

\maketitle

\begin{abstract}
This paper introduces a novel approach to concurrently design dynamic controllers and correlated differential privacy noise in dynamic control systems. An increase in privacy noise increases the system's privacy but adversely affects the system's performance. Our approach optimizes the noise distribution while shaping closed-loop system dynamics such that the privacy noise has the least impact on system performance and the most effect on system privacy. We further add privacy noise to both control input and system output to privatize the system's state for an adversary with access to both communication channels and direct output measurements. 
The study also suggests tailored privacy bounds for different states, providing a comprehensive framework for jointly optimizing system performance and privacy in the context of differential privacy.
\end{abstract}

\section{Introduction}
In today's increasingly interconnected and data-driven world, it has become important for connected entities to share information with each other to work efficiently. This applies not just to individuals, but also to Cyber-Physical Systems (CPS) in various sectors, including industrial control systems, power grids, financial markets, and commercial and military communication networks. This pervasive data sharing has also brought forth heightened concerns regarding system security, safety and privacy. With private entities, government organizations, and adversaries engaging in extensive data collection and analysis, the risk of exposing sensitive information has significantly increased, posing potential harm to both individuals and critical systems. This has led to the development of mechanisms providing different kinds of privacy guarantees. Examples include differential privacy, information-theoretic privacy, and privacy based on secure multiparty computation.

Differential privacy was originally designed to protect the data of individuals in static databases, but its application has expanded to address the privacy challenges posed by dynamic and interconnected data ecosystems, such as Cyber-Physical Systems (CPS) and Internet of Things (IoT) devices \cite{dwork2014algorithmic,dwork2008differential}. At a basic level, a differentially private mechanism ensures that the results of a query remain approximately unchanged if data belonging to any single user, or a single row, in the database are modified \cite{han2018privacy}. Informally, differential privacy makes similar data appear \textit{approximately} indistinguishable from one another \cite{yazdani2022differentially}. The most important feature of differential privacy is its protection from post-processing or its robustness in the presence of side information \cite{hassan2019differential}.  However, there is a price associated with making the system differentially private. Differential privacy works by adding noise to the system which leads to a degradation in system performance both in static and dynamic cases \cite{hassan2019differential,le2013differentially}. 

In recent years, researchers have extended the work on differential privacy for static databases to differential privacy for dynamic filters \cite{le2013differentially}, control and dynamical systems \cite{han2018privacy}, differentially private LQ Control \cite{yazdani2022differentially}, multi-agent formation control \cite{hawkins2022differentially}, and Differentially private distributed constrained optimization \cite{han2016differentially}. 
In differentially private LQ control \cite{yazdani2022differentially}, the authors consider a multi-agent system described using linear system dynamics and add privacy noise such that every agent’s state trajectory is made approximately indistinguishable from all other state trajectories. The paper provides lower and upper bounds on mean square error (MSE) in state estimation for some minimum and maximum privacy noise among agents, where the combined state is estimated using a standard Kalman filter while designing an LQG control for the overall system. The paper further provides guidelines for choosing the privacy level $\epsilon_i$ to bound the MSE in the cloud’s state estimates and further provides the cost of privacy in terms of the increase of the quadratic cost. 
Kawan and Cao \cite{kawano2020design,kawano2018differential} show that the Gaussian mechanism evaluates the maximum eigenvalue of the input observability Gramian and thus the addition of even small noise is enough to make the less input observable Gaussian mechanism highly differentially private.

In this paper, we consider the joint design of dynamic controller and differentially private noise (correlated noise with different variances across channels) such that the system performance loss is minimized for a given privacy metric or the privacy metric is maximized for a given system performance. The idea is to find the optimal privacy noise distribution and simultaneously design the closed-loop system dynamics such that the correlated noises enter the system through channels that have the least impact on system performance and maximize system privacy. It can be understood as the following: the larger noise would only be added through the input/output channels whose effect on the system state has been minimized by designing the closed-loop poles. We design a dynamic controller that directly processes the controller states to generate the control signal instead of first estimating the states and then using it to obtain the control input. We assume a smart adversary that will develop an optimal estimator to generate individual signals for more accurate state estimation by leveraging the additional information about the system dynamics. 
We consider two kinds of adversaries, one that has access to the communication channel and another that has direct access to the measurements. We express a privacy metric in the presence of smart adversaries in terms of uncertainty in the estimation of individual states. We will further design different privacy bounds in different states as some states might need stricter privacy guarantees than other states. This will also be considered for bounding the system performance across different states based on the underlying system.

The organization of the paper can be laid out as follows:
Section~\S \ref{s:DP} provides the necessary background on differential privacy and the notation required for the rest of the paper.  
Section~\S \ref{s:prob} gives the relationship between differential privacy and error in state estimate and then formulates the final design problem. 
Section~\S \ref{s:design} elaborates the system design approach for a general dynamic controller and correlated differential privacy input and output noises and provides the solution as a convex optimization problem.
Section~\S \ref{s:sim} gives simulation results for differential privacy of a networked power distribution system with load frequency control under unknown power demand and \S \ref{s:conc} provides the final concluding remarks along with the future work.











\vspace{-1mm}
\section{Review of Differential Privacy} \label{s:DP}
In this section, we review the basic definitions and define the Gaussian mechanism used to enforce differential privacy in dynamical systems. 
We define the expectation operator by $\mathbb{E}[\cdot]$ and $\mathcal{N}({x}, { Y})$ denotes the Gaussian distribution with mean ${ x}$ and covariance ${Y}$. The diagonal matrix generated from a vector ${ x}$ is denoted as $\operatorname{diag}({x})$ and the block diagonal matrix is denoted as $\operatorname{blkdiag}({ Y}_1, { Y}_2, \cdots ,{Y}_3)$. The symbol $O$ defines a zero matrix with suitable dimensions and $I$ defines the unit matrix of appropriate dimensions. The notations ${ X \succ 0}$ and ${ Y \succeq 0}$ denote the symmetric positive definite and symmetric positive semidefinite matrices, respectively. 

Let us consider agent's state trajectories of the form $x=(x(1), x(2), \ldots)$, where $x(k) \in \mathbb{R}^{n_x}$ and $\|x(k)\|_2<\infty$ for all $k$, and let us denote the set of all such sequences by $x \in {\ell}_2^{n_x}$. Let us define our adjacency relation over ${\ell}_2^{n_x}$.

\begin{definition}
    (Adjacency for trajectories): Let us choose $\beta > 0$ as the adjacency parameter and $v, w \in {\ell}_2^{n_x}$ as two trajectories that are adjacent if $\left\|v-w\right\|_{\ell_2} \leq \beta$. We write $\operatorname{Adj}_{\beta}\left(v,w\right)=1$ if $v$ and $w$ are adjacent, and $\operatorname{Adj}_{\beta}\left(v, w\right)=0$, otherwise.
\end{definition}

This adjacency relation requires that an agent's state trajectory be made approximately indistinguishable within distance $\beta$ from all other state trajectories. Let us consider that the agent's output signal is of dimension $n_y$ at each point in time and is in the set ${\ell}_2^{n_y}$. Next, we define the sensitivity of a dynamical system.

\begin{definition}
(Sensitivity): The p-norm sensitivity of a system $\mathcal{G}$ is the greatest distance between two output trajectories that correspond to adjacent state trajectories:
\begin{align}
\nonumber  \Delta_p\mathcal{G} :=\sup_{x,\tilde{x}|\operatorname{Adj}_B (x,\tilde{x})=1} \|\mathcal{G}(x)-\mathcal{G}(\tilde{x})\|_{p}.
\end{align}
\end{definition}

Now, we define differential privacy for dynamic systems (see \cite{le2013differentially} for a formal construction).

\begin{definition}
(Differential privacy for trajectories): Let $\epsilon>0$ and $\delta \in(0,1 / 2)$ be given. A mechanism $\mathcal{M}(\cdot) \in \ell_2^{n_y} $ is $\left(\epsilon, \delta\right)$-differentially private if, for all adjacent $x, x^{\prime} \in \ell_2^{n_x}$, we have:
\begin{align}
    \nonumber \mathbb{P}\left[\mathcal{M}\left(x\right) \in S\right] \leq e^{\epsilon} \mathbb{P}\left[\mathcal{M}\left(x^{\prime}\right) \in S\right]+\delta \text { for all } S \in \Sigma_2^{n_y} .
\end{align}
\end{definition}

We now define the Gaussian mechanism. 

\begin{lemma}
    (Gaussian mechanism; \cite{le2013differentially}): Let us use privacy parameters $\epsilon>0$ and $\delta \in(0,1 / 2)$ and adjacency parameter $\beta>0$. Let $\mathcal{G}$ denote a dynamical system and $\Delta_2\mathcal{G}$ denote its 2-norm sensitivity. Then the Gaussian mechanism $\mathcal{M} = \mathcal{G}(x)+v^p$ makes the system $\left(\epsilon, \delta\right)$-differentially private with respect to $\operatorname{Adj}_{\beta}$, if $v^p(k) \sim \mathcal{N}\left(0, \sigma^2 I_{n_y}\right)$, and $\sigma \geq  \Delta_2\mathcal{G}~ \beta ~\kappa\left(\delta, \epsilon\right)$, where   $\kappa\left(\delta, \epsilon\right)=\frac{1}{2 \epsilon}\left(K_{\delta}+\sqrt{K_{\delta}^2+2 \epsilon}\right)$, with $K_{\delta}:=\mathcal{Q}^{-1}\left(\delta\right)$,  $\mathcal{Q}$ representing the Gaussian tail integral.
\end{lemma}

\section{Problem Formulation}\label{s:prob}
The main objective of the research is to make the state of the agents differentially private by adding privacy noise while achieving the desired system performance. The privacy noise can be added to (refer \cref{fig:Architecture}):
\begin{itemize}
    \item System output: Add differential privacy noise to the outputs measured by the sensors as $v^p_k$. 

    \item Control input: Add differential privacy noise directly to the control input as $w^p_k$. 
\end{itemize}

\begin{figure}[h!]
    \centering
    \includegraphics[width=.8\linewidth]{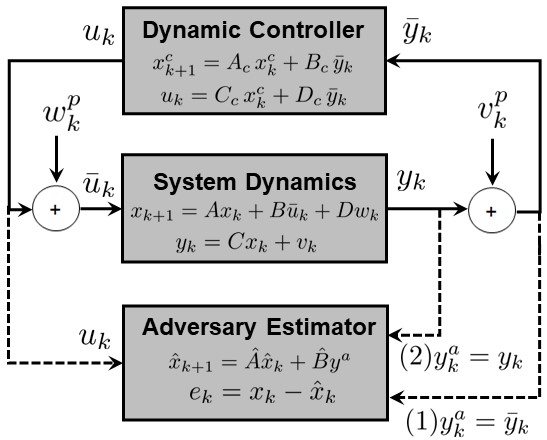}
    \caption{Design architecture for making agent's state differentially private by adding privacy noise to both system inputs and output.}
    \label{fig:Architecture}
\end{figure}

Notice that the control input noise is a physically feasible way to insert privacy noise into the system. Also, the actual privacy noise should be calculated by accounting for the actuator noise present in the system. We further consider two cases based on the capability of the adversary (refer \cref{fig:Architecture}):
\begin{itemize}
    \item In the first case, the adversary listens to the communication between the agents and the centralized controller ((1) $y^a_k = \bar{y}_k$). We add privacy noise to both the outputs measured by the sensor and the control inputs to make the system state differentially private. 

    \item In the second case, the adversary has his own sensors and has direct access to the output of the system ((2) $y^a_k = y_k$), and thus adding output privacy noise alone would not make the system differentially private.
    Although there is no benefit in adding privacy noise to the sensor side, we still add it and expect the design to remove the privacy noise on the output side.
\end{itemize}


Let us consider a discrete-time linear time-invariant (LTI) system, along with the addition of output privacy noise and control input privacy noise, described by the following state-space representation:
\begin{align}
x_{k+1} &= A x_k +B (u_k + w^p_k) + D w_k ,\label{state_eqn1} \\
y_k &=C x_k + v_k ,  \label{output_eqn1} \\
\bar{y}_k &=C x_k + v_k + v^p_k,  \label{output_eqn2} \\
z_k &= C_z x_k,  \label{meas_eqn2}
\end{align}
where $x_k \in \mathbb{R}^{n_x}$ is the state of the system at time $k$, $u_k \in \mathbb{R}^{n_u}$ is the control vector at time $k$. The initial state vector $x_0$ and the process noise at time $k$, $w_k$, are assumed to be independent random variables. In particular, $w_k \sim \mathcal{N}(\mathbf {0},W)$, $\forall k$, with $W\in \mathbb R^{n_w\times n_w}$ to be known and fixed covariance matrix. 
The output of the system $y_k \in \mathbb{R}^{n_y}$, is measured by a sensor network with sensor noise modeled as another independent Gaussian random variable $v_k \sim \mathcal{N}(\mathbf{0},V)$, $\forall k$ with $V\in \mathbb R^{n_v\times n_v}$ to be the known and fixed covariance matrix. The performance variable $z_k \in \mathbb{R}^{n_z}$ defines the variables of interest for the system.

The control input privacy noise, $w^p_k$, and output privacy noise, $v^p_k$, are modeled as random variables, $w^p_k \sim \mathcal{N}(\mathbf {0},W^p)$ and $v^p_k \sim \mathcal{N}(\mathbf{0},V^p)$, $\forall k$, with $W^p$ and $V^p$ being the covariance matrix, representing the strength of the added noise. We further define the inverse of the respective noise covariance matrices as:
$$\Gamma_w = W^{p^{-1}}, ~~ \Gamma_v = V^{p^{-1}}.$$

We assume that the adversary is smart and has full information about the system dynamics, i.e., knowledge of system matrices, $A,B,C,D$. Moreover, the adversary will design an optimal estimator to estimate the system state $\hat{x}_k$ using a general estimator of the form:
\begin{align}
    \hat{x}_{k+1} &=\hat{A} \hat{x}_{k}+ Bu_k+\hat{B} y^a_{k},  \label{e:est1}\\
    e_k &= x_k - \hat{x}_k, \label{e:est2}
\end{align}
such that the error in the state estimate $e_k$ is minimized.

\subsection{Relationship between differential privacy and error in state estimates due to control input and output privacy noise}
In this subsection, we show that the differential privacy of the system state can be represented by the error in the estimation of the state while using the optimal state estimator by the adversary of the form \cref{e:est1,e:est2}.  In particular, the covariance of the state error estimates $E_k = \mathbb{E}[e_k e_k^T]$ can be used as a metric to quantify differential privacy, and an increase in error covariance results in an increase in $\left(\epsilon, \delta \right)$-differentially privacy. 

\begin{lemma}
    (Gaussian mechanism for dynamical system; \cite{le2013differentially}): Let $\mathcal{G}$ denote an LTI dynamical system and $\|\mathcal{G}\|_\infty < \infty$ and let us use privacy parameters $\epsilon, \delta >0$. Then the Gaussian mechanism $\mathcal{M} u = \mathcal{G} u + w^p$, where $w^p$ is a Gaussian noise with $w^p \sim \mathcal{N}\left(0, \sigma^2 I_{n_y}\right)$, and $\sigma \geq  \beta ~\kappa\left(\delta, \epsilon\right) \|\mathcal{G}\|_\infty $,
    makes the system $\left(\epsilon, \delta\right)$-differentially private with respect to $\operatorname{Adj}_{\beta}$ in $u$ , i.e., $\|u-u'\|_2 \leq \beta$ with $\beta>0$. 
\end{lemma}

\begin{remark}
    The above lemma allows us to make the control input differential private by directly adding the noise to the control inputs when the output is queried and obtained by passing through an LTI dynamical system. In this paper, we want to make the system state differentially private by adding privacy noise to both system inputs and outputs.
\end{remark}

Yazdani et.~al.~\cite{yazdani2022differentially} used the level of privacy to calculate the impact on estimation error and showed the relationship between the privacy noise and the trace of covariance of the state error estimates $\operatorname{tr}(E_k)$ where the state estimates are calculated using a Kalman filter. However, another way to look at the impact of differential privacy from the point of view of an adversary is to hinder his capability to estimate the state trajectories accurately. So if an adversary designs an optimal estimator, the error in estimating state trajectories can be used as a metric of differential privacy. Next, we expand on the results generated in \cite{yazdani2022differentially} to quantify standard ($\epsilon,\delta$)-differential privacy as the error in adversary's state estimates due to both control input and output privacy noise.

\begin{lemma}
For the given dynamical systems in (Eq.~(\ref{state_eqn1})-(\ref{output_eqn2})) with both control input privacy noise $w^p_k \sim \mathcal{N}(\mathbf {0},W^p)$ and output privacy noise $v^p_k \sim \mathcal{N}(\mathbf{0},\sigma^2 I_{n_y}-V)$, with $\sigma = \bar{S}(C) \beta k(\delta,\epsilon)$, and for a given adjacency $\|x-x'\|_2 \leq \beta$ with $\beta>0$, if the states are ($\epsilon,\delta$)-differential private with $\delta \in [10^{-5},10^{-1}]$ and
\begin{align*}
 \epsilon &\leq \left( \frac{\bar{S}(C)^2 \beta^2 (n_x-\operatorname{tr} \left(\underbar{E}) \underline{\lambda}(\Psi)^{-1}\right)}{ \operatorname{tr}(\underbar{E}) C_u^2} \right)^{1/2}, \\
 \Psi &= DWD^T+BW^p B^T, \\
 V^p &= \sigma^2 I_{n_y}-V
\end{align*}
where $\bar{S}(\cdot)$ represents the maximum singular value, $\underline{\lambda}(\cdot)$ represents the smallest eigenvalue of the matrix, and $C_u$ is the value of $C$ corresponding to the index for which the diagonal element of $C^T(V+V^p)^{-1}C)$ is maximum, then the state error estimate is lower bounded by $\operatorname{tr}(\underbar{E})$ with $\mathbb{E}[e_k e_k^T]>\underbar{{E}}$.  
\end{lemma}
\begin{proof}
Here we consider both control input privacy noise and output privacy noise along with the already present process and measurement noise. Thus the equation for \textit{a priori} state error covariance follows: $$\Sigma = A(\Sigma^{-1}+C^T(V+V^p)^{-1}C)^{-1}A^T + DWD^T+BW^pW^T,$$ and for \textit{a posteriori} state error covariance follows: $$\overline{\Sigma} = (\Sigma^{-1}+C^T(V+V^p)^{-1}C)^{-1}.$$ 
After that, it follows directly from (Theorem 2 of \cite{yazdani2022differentially}) where we consider a single agent with state dimension $n_x$ instead of the multi-agent case.
\end{proof}

The above result is used to show that enforcing differential privacy to the systems' state ensures a lower bound on state estimation error. Please note that the above result provides a necessary condition for differential privacy based on the lower bound on state estimation error. More work is needed to find the bounds for sufficiency.


\begin{remark}
    Please note that both the performance norm $\mathbb{E}[z_k z_k^T]$ and the estimator error covariance $\mathbb{E}[e_k e_k^T]$ increases with an increase in output and control input privacy noises, and our objective is to find the optimum noise level along with the controller and estimator to minimize the performance norm for a given error covariance.
\end{remark}

\textbf{Main Design Problem Formulation:} Design the strength of privacy noises, $W^p (\Gamma_w)$ and $V^p(\Gamma_v)$, and an optimal state estimator of the form \cref{e:est1,e:est2}, and a general linear dynamic controller of the form:
\begin{align}
x^{c}_{k+1} &= A_c\, x^{c}_{k} +B_c\, \bar{y}_k, \label{e:Cntrl_dyn1}\\
u_k &= C_c\, x^{c}_k + D_c\, \bar{y}_k. \label{e:Cntrl_dyn2}
\end{align}
such that the state error covariance $\mathbb{E}[e_k e_k^T]$ is maximized while closed-loop system performance is bounded $\mathbb{E}[z_k z_k^T]\leq\bar{\mathbf{Z}}$. 
\begin{align}
\begin{array}{lll}
        \max_{\{A_c, B_c,C_c, D_c, \Gamma_w, \Gamma_v, \hat{A}, \hat{B}\}} & \operatorname{tr}(\mathbb{E}[e_k e_k^T]) \\
        \textit{s.t.} & \mathbb{E}[z_k z_k^T]\leq\bar{\mathbf{Z}} .
\end{array}
\end{align}

Another problem of interest can be to minimize the closed-loop system performance $\mathbb{E}[z_k z_k^T]$ while lower bounding the state error covariance $\mathbb{E}[e_k e_k^T]>\underbar{{E}}$ for some given $\underbar{{E}}$, i.e., to have higher differential privacy than some specified limit.
\begin{align}
\begin{array}{lll}
        \min_{\{A_c, B_c,C_c, D_c, \Gamma_w, \Gamma_v, \hat{A}, \hat{B}\}} & \operatorname{tr}(\mathbb{E}[z_k z_k^T]) \\
        \textit{s.t.} & \mathbb{E}[e_k e_k^T]\geq\mathbf{\underbar{E}}.
\end{array}
\end{align}

For the two cases that we discussed based on the capabilities of the adversary, the information available to the estimator would change from (1) $y^a_k = \bar{y}_k$ to (2) $y^a_k = y_k$.

\begin{remark}
    Notice that the estimator design from the point of view of the adversary is general and can be used to simultaneously design the estimator with the privacy noise for the case of open loop system dynamics also, i.e. with $u_k = 0$. 
\end{remark}

\section{Final Design Solution Development}\label{s:design}
In this section, we develop frameworks for the co-design of input and output privacy noise with a dynamic feedback controller; and the co-design of input and output privacy noise with an optimal estimator. We further provide the final design algorithm for the two cases of adversarial capabilities. In both cases, we formulate the problem such that the controller gets the output signal with added privacy noise $v^p$ and let the optimization solve the optimal privacy noise.

\subsection{Adversary with access to communication channels}
For the case where the adversary listens to the noisy output passed through the communication channel $y^a_k = \bar{y}_k = C x_k + v_k + v^p_k$, the final design problem can be solved using the following result. 

\begin{theorem}
    For the dynamical system given in \cref{state_eqn1,output_eqn2} with adversary listening through the communication channel, and to maximize differential privacy for a fixed performance bound, the optimal design solution with privacy noises, $W^p(\Gamma_w)$ and $V^p(\Gamma_v)$, an optimal state estimator of the form \cref{e:est1,e:est2}, and a general linear dynamic controller of the form \cref{e:Cntrl_dyn1,e:Cntrl_dyn2}, can be solved as a convex optimization problem using the following LMIs:
\begin{align*}
  \textit{minimize}_{\{A_c, B_c,C_c, \Gamma_w, \Gamma_v, \hat{A}, \hat{B}\}}  ~ \operatorname{trace}\left(\Gamma_w+\Gamma_v \right)  ,
\end{align*}
\begin{equation}
\begin{bmatrix}
X & I & A X+B L & A  & D &  B & O & O  \\
(\cdot)^{T} & Y & Q & Y A+F C  & Y D & Y B & F & F  \\
(\cdot)^{T} & (\cdot)^{T} & X & I  & O & O & O & O \\
(\cdot)^{T} & (\cdot)^{T} & (\cdot)^{T} & Y  & O & O & O & O \\
(\cdot)^{T}  & (\cdot)^{T} & (\cdot)^{T} & (\cdot)^{T} & W^{-1} & O & O & O\\
(\cdot)^{T} &  (\cdot)^{T} & (\cdot)^{T} & (\cdot)^{T} & (\cdot)^{T} & \Gamma_w & O & O \\
(\cdot)^{T} &  (\cdot)^{T} & (\cdot)^{T} & (\cdot)^{T} & (\cdot)^{T} & (\cdot)^{T} & V^{-1}  & O\\
(\cdot)^{T} &  (\cdot)^{T} & (\cdot)^{T} & (\cdot)^{T} & (\cdot)^{T} & (\cdot)^{T} & (\cdot)^{T} & \Gamma_v 
\end{bmatrix}>O, \label{e:Closed_perf1}
\end{equation}
\begin{equation}
\begin{bmatrix}
\bar{\mathbf Z}  & C_z X &  C_z  \\
(\cdot)^{T} & X & I \\
(\cdot)^{T} & (\cdot)^{T} & Y  
\end{bmatrix}>O, \label{e:Closed_perf2}
\end{equation}
\begin{equation}
    \begin{bmatrix}
\hat{X} & I & A \hat{X}  & A  & D &  B & O & O  \\
(\cdot)^{T} & \hat{Y} & \hat{Q} & \hat{Y} A + \hat{F} C  & \hat{Y} D & \hat{Y} B & \hat{F} & \hat{F}  \\
(\cdot)^{T} & (\cdot)^{T} & \hat{X} & I  & O & O & O & O \\
(\cdot)^{T} & (\cdot)^{T} & (\cdot)^{T} & \hat{Y}  & O & O & O & O \\
(\cdot)^{T}  & (\cdot)^{T} & (\cdot)^{T} & (\cdot)^{T} & W^{-1} & O & O & O\\
(\cdot)^{T} &  (\cdot)^{T} & (\cdot)^{T} & (\cdot)^{T} & (\cdot)^{T} & \Gamma_w & O & O \\
(\cdot)^{T} &  (\cdot)^{T} & (\cdot)^{T} & (\cdot)^{T} & (\cdot)^{T} & (\cdot)^{T} & V^{-1}  & O\\
(\cdot)^{T} &  (\cdot)^{T} & (\cdot)^{T} & (\cdot)^{T} & (\cdot)^{T} & (\cdot)^{T} & (\cdot)^{T} & \Gamma_v 
\end{bmatrix}>O, \label{e:stable_est1}
\end{equation}
\begin{equation}
\begin{bmatrix}
\bar{\mathbf E}  & \hat{X}-\hat{U} &  I \\
(\cdot)^{T} & \hat{X} & I \\
(\cdot)^{T} & (\cdot)^{T} & \hat{Y} 
\end{bmatrix}>O,\label{e:stable_est2}
\end{equation}
\noindent and finally, the optimal estimator and dynamic controller can be calculated as:
\begin{align}
    \hat{A} &= \hat{S}^{-1}(\hat{Q}-\hat{Y}A\hat{X}-\hat{F} C \hat{X})\hat{U}^{-1}, \label{e:stable_est_gain1}\\
    \hat{B} &= \hat{S}^{-1}\hat{F}, \label{e:stable_est_gain2}
\end{align}
\begin{equation}
\begin{split}
\begin{bmatrix}
A_{c} & B_{c} \\
C_{c} & D_{c}
 \end{bmatrix} = & \begin{bmatrix}
S^{-1} & -S^{-1} Y B \\
O & I
 \end{bmatrix}  \begin{bmatrix}
Q-Y A X & F \\
L & O
 \end{bmatrix}  \\ & \cdot \begin{bmatrix}
U^{-1} & O \\
-C X U^{-1} & I
 \end{bmatrix}. 
\end{split} \label{e:RLFQ}
\end{equation}
\end{theorem}
\begin{proof}
    We first design the dynamic controller of the form \cref{e:Cntrl_dyn1,e:Cntrl_dyn2}. We assume the direct feedforward term in the dynamic controller to be zero $D_c = O$ for the bounded control input covariance. Using the above compensator, the closed-loop system dynamics can be written using the augmented state vector $\mathbf{x}^T := \begin{bmatrix} x^T & {x^c}^T  \end{bmatrix}$ with augmented process noise $\mathbf{w}^T := \begin{bmatrix}  w^T & w^{p^T} & v^T & v^{p^T}\end{bmatrix}$ as:
\begin{align}
    \mathbf{x}_{k+1} &= \mathbf A \, \mathbf x_k +\mathbf {B}\, \mathbf{w}_{k}, \label{eq:closedA}\\
{z}_{k} &= \mathbf C\, \mathbf{x}_{k},\label{eq:closedC}
\end{align}
\begin{align*}
\text{where}~~   \mathbf A = \begin{bmatrix} A  & BC_c \\ B_c C & A_c \end{bmatrix}, ~ 
    \mathbf B = \begin{bmatrix}  D & B & O & O\\  O & O & B_c & B_c \end{bmatrix}, 
\end{align*}
$\mathbf C = \begin{bmatrix} C_z & O \end{bmatrix}$ and $\mathbf w_{k}\sim \mathcal N(\mathbf 0, \mathbf W )$, where 
\begin{align}
    \mathbf W^{-1} = \begin{bmatrix}  W^{-1} & O & O & O \\ O  & \Gamma_w & O & O\\ O & O & V^{-1} & O\\  O & O& O& \Gamma_v \end{bmatrix}.
\end{align}

It is a standard result that the above closed loop system is stable and a steady-state state covariance matrix ($\mathbf X > 0 $) exists, if:
\begin{equation} 
\mathbf A\mathbf X\mathbf A^T+\mathbf B\mathbf W\mathbf B^T < \mathbf X, 
\end{equation}
which using Schur's complement gives: 
\begin{equation}
\begin{bmatrix} 
\mathbf X & \mathbf A\mathbf X & \mathbf B \\
(\cdot)^{T}& \mathbf X & O \\
(\cdot)^{T} & (\cdot)^{T} & \mathbf W^{-1} \end{bmatrix} > O, \label{stable_ineq11} 
\end{equation}
where $(\cdot)^{T}$ represents the corresponding transpose of the symmetric block. It is straightforward to show that the performance covariance can be bounded as:
$\mathbf C\mathbf X \mathbf C^T  < \bar{\mathbf Z},$ which can be written as:
\begin{equation}
    \begin{bmatrix}
   \bar{\mathbf Z} & \mathbf C \mathbf X  \\
(\cdot)^{T} & \mathbf X \end{bmatrix} > O .\label{e:out_bound}
\end{equation}

\noindent Notice that the constraint in Eqn.~(\ref{stable_ineq11}) is not an LMI. We need to perform congruence transformation and change of variables to convert them to LMIs \cite{Scherer_LMI_1997,Integrating2021MSSP}. Let us define and partition the matrix as:
\begin{equation*}
    \mathbf X \triangleq  \begin{bmatrix}
X & U^{T} \\
U & \hat{X}
\end{bmatrix} , \quad \mathbf X^{-1} \triangleq \begin{bmatrix}
Y & S \\
S^{T} & \hat{Y}\end{bmatrix},
\end{equation*}
and the transformation matrix \
\begin{equation*}
    \mathbf{T} \triangleq \begin{bmatrix}
    I & Y \\ O & S^T 
    \end{bmatrix}
\end{equation*}
and associated congruence transformation matrices
\begin{equation*}
    \mathbb T \triangleq \begin{bmatrix}
    \mathbf{T}  & O & O \\
O & \mathbf{T} & O \\
O & O & I
    \end{bmatrix}, \tilde{\mathbb T} \triangleq \begin{bmatrix}
    I & O  \\
O & \mathbf{T} \end{bmatrix}.
\end{equation*}
\noindent Applying $\{\mathbb T, \tilde{\mathbb T}\}$ to Eqn.~(\ref{stable_ineq11}) and Eqn.~(\ref{e:out_bound}), we obtain:
\begin{equation}\label{eq: transformation1}
  \mathbb T^T  \begin{bmatrix}
    \mathbf X & \mathbf A \mathbf X & \mathbf B \\
(\cdot)^{T} & \mathbf X& O \\
(\cdot)^{T} & (\cdot)^{T} & \mathbf W^{-1}
    \end{bmatrix} \mathbb T > O,
\end{equation}

\begin{equation} \label{eq: transformation2}
    \tilde{\mathbb T}^T \begin{bmatrix}
   \bar{\mathbf Z} & \mathbf C \mathbf X  \\
(\cdot)^{T} & \mathbf X 
    \end{bmatrix} \tilde{\mathbb T} > O .
\end{equation} 


Expansion of \cref{eq: transformation1,eq: transformation2} under an appropriate change of variables leads to a set of LMIs \cref{e:Closed_perf1,e:Closed_perf2} that do not depend on $S$ or $U$. Once the $X, Y$ are obtained, matrices $S$ and $U$ need to be constructed using:
\begin{equation}
Y X+S U=I, \label{eq:YXSU}
\end{equation}
and, a handy choice of $U$ and $S$ satisfying \cref{eq:YXSU} is $S=Y$, then $U=Y^{-1}-X$.
\noindent Notice that when the controller has the same order as the plant, $S$ and $U$ are square and non-singular matrices, in which case the controller gain matrices can be calculated using \cref{e:RLFQ}.


Although the original problem was to maximize the $\mathbb{E}[e_k e_k^T]$ to increase the differential privacy, but an increase in error covariance can also result from suboptimal estimator gains. However, as we consider a smart adversary, who would always design an optimal estimator, we update the design problem to maximize the privacy noises while bounding the error covariance. 

Let us design the estimator to bound the error covariance for different states with the estimator dynamics given as:
\begin{align}
    \hat{x}_{k+1} =\hat{A} \hat{x}_{k} + B u_k +\hat{B} y^a_{k}.
\end{align}

Combining the above estimator and the underlying dynamics given in \cref{state_eqn1,output_eqn1}, the combined dynamics can be written using the augmented state vector $\mathbf{\hat{x}}^T := \begin{bmatrix} x^T & {\hat{x}}^T  \end{bmatrix}$ as:
\begin{align}
    \mathbf{\hat{x}}_{k+1} &= \mathbf{\hat{A}} \, \mathbf{\hat{x}_k} +\mathbf{\hat{B}}\, \mathbf{w}_{k}, \label{eq:closedEstA}
\end{align}

\begin{align}
\text{where}~    \mathbf{\hat{A}} = \begin{bmatrix} A  & O \\ \hat{B} C & \hat{A} \end{bmatrix},  ~~
    \mathbf{\hat{B}} = \begin{bmatrix}   D & B & O & O\\  O & O &  \hat{B}  &  \hat{B} \end{bmatrix},
\end{align}
and the error in estimation can be written as:
\begin{align}
\mathbf{e}_{k} =  \mathbf{\hat{C}} \mathbf{\hat{x}}_{k}, ~~~ \mathbf{\hat{C}} = \begin{bmatrix} I&-I\end{bmatrix}, \label{eq:closedEstC}
\end{align}

Similar to the previous development for the existence of the steady-state state covariance matrix ($\mathbf{\hat{X}} > 0 $) and stability of the system, we write:
\begin{equation}
\begin{bmatrix} 
\mathbf{\hat{X}} & \mathbf{\hat{A}} \mathbf{\hat{X}} & \mathbf{\hat{B}} \\
(\cdot)^{T}& \mathbf{\hat{X}} & O \\
(\cdot)^{T} & (\cdot)^{T} & \mathbf W^{-1} \end{bmatrix} > O, \label{e:Est_stable}
\end{equation}
and the error covariance can be bounded as:
\begin{equation}
    \begin{bmatrix}
   \bar{\mathbf E} & \mathbf{\hat{C}} \mathbf{\hat{X}}  \\ (\cdot)^{T} & \mathbf{\hat{X}} \end{bmatrix} > O .
\end{equation}

\noindent Again noticing that the constraint in Eqn.~(\ref{e:Est_stable}) is not an LMI, we follow a similar procedure, by partitioning the state covariance matrix and performing the congruence transformation to obtain the LMIs given in \cref{e:stable_est1,e:stable_est2}. Once the $\hat{X}, \hat{Y},\hat{U},\hat{F},\hat{Q} $ are obtained, matrix $\hat{S}$ can be constructed using:
\begin{equation}
\hat{S}=(I-\hat{Y} \hat{X})\hat{U}^{-1}, 
\end{equation}
and the estimator matrices can be constructed using \cref{e:stable_est_gain1,e:stable_est_gain2}.
\end{proof}

\subsection{Adversary with direct access to measurements}
For the case where the adversary uses his own sensors to measure the system output $y^a_k = y_k = C x_k + v_k$, the output privacy noise will not help in privatizing the system and the final design problem will change as follows.
\begin{theorem}
    For the dynamical system given in \cref{state_eqn1,output_eqn2} with an adversary using his own sensors to measure the system output, and to maximize differential privacy for a fixed performance bound, the optimal design solution with privacy noises, $W^p(\Gamma_w)$ and $V^p(\Gamma_v)$, an optimal state estimator (\cref{e:est1,e:est2}), and a general linear dynamic controller (\cref{e:Cntrl_dyn1,e:Cntrl_dyn2}), can be solved as a convex optimization problem using the following LMIs:
\begin{align*}
  \textit{minimize}_{\{A_c, B_c,C_c, \Gamma_w, \Gamma_v, \hat{A}, \hat{B}\}}  ~ \operatorname{trace}\left(\Gamma_w \right)  ,
\end{align*}
~~~~~~~~~~ $\mathbb{E}[z_k z_k^T]<\bar{\mathbf{Z}} \rightarrow$ (\cref{e:Closed_perf1,e:Closed_perf2} (LMIs),
\begin{equation}
    \begin{bmatrix}
\hat{X} & I & A \hat{X}  & A  & D &  B & O   \\
(\cdot)^{T} & \hat{Y} & \hat{Q} & \hat{Y} A + \hat{F} C  & \hat{Y} D & \hat{Y} B & \hat{F}  \\
(\cdot)^{T} & (\cdot)^{T} & \hat{X} & I  & O & O & O  \\
(\cdot)^{T} & (\cdot)^{T} & (\cdot)^{T} & \hat{Y}  & O & O & O  \\
(\cdot)^{T}  & (\cdot)^{T} & (\cdot)^{T} & (\cdot)^{T} & W^{-1} & O & O \\
(\cdot)^{T} &  (\cdot)^{T} & (\cdot)^{T} & (\cdot)^{T} & (\cdot)^{T} & \Gamma_w & O  \\
(\cdot)^{T} &  (\cdot)^{T} & (\cdot)^{T} & (\cdot)^{T} & (\cdot)^{T} & (\cdot)^{T} & V^{-1}  
\end{bmatrix}>O,
\end{equation}
\begin{equation}
\begin{bmatrix}
\bar{\mathbf E}  & \hat{X}-\hat{U} &  I \\
(\cdot)^{T} & \hat{X} & I \\
(\cdot)^{T} & (\cdot)^{T} & \hat{Y} 
\end{bmatrix}>O,
\end{equation}
\noindent and finally, the optimal estimator and dynamic controller can be calculated as \cref{e:stable_est_gain1,e:stable_est_gain2} and \cref{e:RLFQ}.
\end{theorem}
\begin{proof}
    The design solution for the dynamic controller is the same as the previous solution. 
    The proof follows similarly to the previous design solution with $y^a_k = \bar{y}_k$ replaced with $y^a_k = {y}_k$ and thus the derivation for the equation for the estimator follows naturally from \cref{e:stable_est1,e:stable_est2} to the above-mentioned results.
\end{proof}

\begin{remark}
    Notice that the design solution presented in Theorem 1 and Theorem 2 is overparameterized, i.e., can help many different controllers/estimators to obtain the same result with different realizations of $YX+SU=I$ and that can further serve as an optimization domain for some other higher level objective function.
\end{remark}

\subsection{Estimator for unstable systems}
The discrete estimator design approach presented earlier is not valid for unstable systems as \cref{e:Est_stable} requires the system matrix $\mathbf{\hat{A}}$ to be stable for a  valid positive definite solution for $\mathbf{\hat{X}}>O$. Thus, we now present the results for unstable system dynamics which restricts the estimator design to: 
\begin{align*}
    \hat{x}_{k+1} =\hat{A} \hat{x}_{k}+ B u_k + \hat{B} y^a_{k}, ~\text{where} ~ \hat{A} = A-\hat{B}C,
\end{align*}
and thus $\hat{B}$ is the only design variable for the estimator. Notice that the control input will cancel out in the estimator as the adversary also has direct access to it. Combining the above estimator and the underlying dynamics given in \cref{state_eqn1,output_eqn1}, the error dynamics can be written as:
\begin{align}
\nonumber    e_{k+1} = (A-\hat{B}C) e_{k}+ D w_k + B w^p_k - \hat{B} v_{k} - \hat{B} v^p_{k}.
\end{align}

Now for the unstable dynamical system, the system design problem can be solved using the following results. Notice that the approach can also be used to only design an estimator and input/output privacy noises for the case of an open-loop unstable dynamical process.
\begin{corollary}
   For the unstable dynamical system given in \cref{state_eqn1,output_eqn2} with an adversary using his own sensors to measure the output, and to maximize differential privacy for a fixed performance bound, the optimal design solution with privacy noises, $W^p(\Gamma_w)$ and $V^p(\Gamma_v)$, an optimal state estimator of the form \cref{e:est1,e:est2}, and a general linear dynamic controller of the form \cref{e:Cntrl_dyn1,e:Cntrl_dyn2}, can be solved as a convex optimization problem using the following LMIs:
\begin{align*}
  \textit{minimize}_{\{A_c, B_c,C_c, \Gamma_w, \Gamma_v, \hat{B}\}}  ~ \operatorname{trace}\left(\Gamma_w + \Gamma_v \right)  ,
\end{align*}
~~~~~~~~~~ $\mathbb{E}[z_k z_k^T]<\bar{\mathbf{Z}} \rightarrow$ (\cref{e:Closed_perf1,e:Closed_perf2} (LMIs),
\begin{align}
&~~~~~~~~~~~~~\begin{bmatrix}
    \bar{\mathbf E}  & I\\I & \hat{Y}
\end{bmatrix} > O, \label{e:Est_per1} \\
&\begin{bmatrix}
    \hat{Y} & \hat{Y}A-\hat{Z} C & \hat{Y}D & \hat{Y}B & \hat{Z} &  \hat{Z}\\
    (\cdot)^T & \hat{Y} &O&O&O&O\\
    (\cdot)^T & (\cdot)^T & W^{-1} &O &O&O\\
    (\cdot)^T & (\cdot)^T & (\cdot)^T & \Gamma_w &O&O\\
    (\cdot)^T & (\cdot)^T & (\cdot)^T & (\cdot)^T& V^{-1} &O\\
    (\cdot)^T & (\cdot)^T & (\cdot)^T & (\cdot)^T&(\cdot)^T& \Gamma_v 
    \end{bmatrix} > O,\label{e:Est_per2}
\end{align}
\noindent and finally, the dynamic controller can be calculated as \cref{e:RLFQ} and optimal estimator as: $\hat{B} = \hat{Y}^{-1} \hat{Z}.$
\end{corollary}
\begin{proof}
We use the standard result that the error dynamics is stable and a steady-state error covariance matrix (${E} > 0 $) is bounded, if: ~~~~~~~~~~~~ $E<\bar{\mathbf E}$
\begin{align*}
     (A-\hat{B}C) E (A-\hat{B}C)^T  & + D W D^T    + B W^p B^T  \\ ~~~~~~~~~~~~~~~~~~~~~~ + \hat{B} V \hat{B}^T+ &\hat{B} V^p \hat{B}^T  < E,
\end{align*}
which can be written using Schur's complement:
 \begin{align*}
&~~~~~~~~~~~~~~~~~~~~~~~~ E < \bar{\mathbf E} , \\
& \begin{bmatrix}
    E & A-\hat{B} C & D & B & \hat{B} &  \hat{B} \\
    (\cdot)^T & E^{-1} &O&O&O&O\\
    (\cdot)^T & (\cdot)^T & W^{-1} &O &O&O\\
    (\cdot)^T & (\cdot)^T & (\cdot)^T & \Gamma_w &O&O\\
    (\cdot)^T & (\cdot)^T & (\cdot)^T & (\cdot)^T& V^{-1} &O\\
    (\cdot)^T & (\cdot)^T & (\cdot)^T & (\cdot)^T&(\cdot)^T& \Gamma_v 
    \end{bmatrix} > O.
\end{align*}

Then, we defined $\hat{Y} = E^{-1}$ and multiplied both sides by matrix $[\operatorname{blkdiag}(\hat{Y},I,I,I)]$ to apply congruence transformation and finally defined $\hat{Z} = \hat{Y}\hat{B}$ to obtain the LMIs in $\hat{Y},\hat{Z},\Gamma_w$ and $\Gamma_v$ as \cref{e:Est_per1,e:Est_per2}.
\end{proof}

\begin{corollary}
        Similar to Corollary 1, but with an adversary using his own sensors to measure the system output, the optimal design problem can be solved as a convex optimization problem using the following LMIs:
        \vspace{-1mm}
\begin{align*}
  \textit{minimize}_{\{A_c, B_c,C_c, \Gamma_w, \Gamma_v, \hat{B}\}}  ~ \operatorname{trace}\left(\Gamma_w \right)  ,
\end{align*}
~~~~~~~~~~ $\mathbb{E}[z_k z_k^T]<\bar{\mathbf{Z}} \rightarrow$ (\cref{e:Closed_perf1,e:Closed_perf2} (LMIs),
\begin{align}
&~~~~~~~~~~~~~\begin{bmatrix}
    \bar{\mathbf E}  & I\\I & \hat{Y}
\end{bmatrix} > O, \label{e:Est_per11} \\
&\begin{bmatrix}
    \hat{Y} & \hat{Y}A-\hat{Z} C & \hat{Y}D & \hat{Y}B & \hat{Z} \\
    (\cdot)^T & \hat{Y} &O&O&O\\
    (\cdot)^T & (\cdot)^T & W^{-1}  &O&O\\
    (\cdot)^T & (\cdot)^T & (\cdot)^T & \Gamma_w &O\\
    (\cdot)^T & (\cdot)^T & (\cdot)^T & (\cdot)^T& V^{-1}  
    \end{bmatrix} > O.\label{e:Est_per21}
\end{align}
\noindent and finally, the dynamic controller can be calculated as \cref{e:RLFQ} and optimal estimator as: $\hat{B} = \hat{Y}^{-1} \hat{Z}.$
\end{corollary}

\section{Simulation Results}\label{s:sim}
The Load Frequency Control (LFC) system maintains a balanced power distribution across different regions by continuously aligning energy demand with generation. LFC involves the transmission of data from remote areas to a central control center, and back to the power production facilities. This communication process in power grids has well-known privacy concerns and thus becomes the motivation for our example \cite{yazdani2022differentially,IEEE_Security}. We illustrate the performance of the proposed architecture on a connected four-area network which is obtained from a network-reduced IEEE New England 39-bus system \cite{bevrani2014robust}. We consider a lossless, connected, and network-reduced power system with each generator modeled by the following equation \cite{bevrani2014robust}:
\begin{align*}
    \dot{\theta}_i(t) &= \omega_i(t),\\
    M_i \dot{\omega}_i(t)  &= -D_i \omega_i(t)-\sum_{j=1}^n B_{ij} V_i V_j \sin(\theta_i(t)-\theta_j(t)) \\ &~~~~~~~~~~~~~~~~~~ +P_{t_i}(t)+w_{p_i}(t),\\
    \tau_{t_i} \dot{P}_{t_i}(t) &= -P_{t_i}(t)-R_i^{-1}\omega_i(t)+u_i(t),
\end{align*}
where $\theta_i(t)$ is the generator rotor angles w.r.t a synchronously rotating reference axis, $\omega_i(t)$ is the frequency deviation w.r.t a synchronous frequency which is $120\pi$ for a 60 Hz system, $M_i$ represents the inertia, $D_i$ represents the damping matrix, $w_{p_i}$ represents the unknown power demand modeled as disturbance, $R_i$ represents the frequency-droop, and $P_{t_i}(t)$ and $\tau_{t_i}$ are the turbine power and time constants, respectively \cite{bevrani2014robust}. 
We linearize the generator model and define the state with four-area network system as:
$$\dot{x}(t) = A_c x(t) + B_c u(t) + D w_p(t),$$
$$x = [\theta_1~\omega_1~P_{t_1} ~\theta_2~\omega_2~P_{t_2}~ \theta_3~\omega_3~P_{t_3}~ \theta_4~\omega_4~P_{t_4}]^T,$$
$$u(t) = [u_1(t)~u_2(t)~u_3(t)~u_4(t)]^T,$$
with the parameters for the networked system given in \cref{t:parameters}. Finally, we discretize the system dynamics with $A = e^{A_c \Delta t}$ and $B = \int_0^{\Delta t} e^{A_c\tau} B_c d \tau$, where $\Delta t$ is the sampling period.

\begin{table}[h!]
    \centering
    \vspace{-2mm}
    \caption{Network Parameters}
    \label{t:parameters}
    \begin{tabular}{c|cccc}
\hline Parameters & Area 1 & Area 2 & Area 3 & Area 4 \\
\hline$M_i$ & 0.1667 & 0.2222 & 0.16 & 0.1304 \\
$D_i$ & 0.0083 & 0.0088 & 0.0080 & 0.0088 \\
$R_i$ & $2.4$ & $2.7 $ & $2.5 $ & $2 $ \\
$\tau_t$ & 0.3 & 0.33 & 0.35 & 0.375 \\
\hline
\end{tabular}
    \vspace{-2mm}
\end{table}
\begin{figure}[h!]
    \centering
    \includegraphics[width=.5\linewidth]{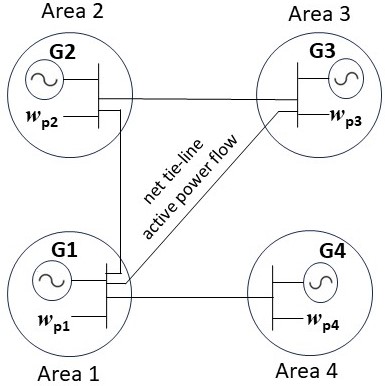}
    \caption{Interconnected four‐area power distribution system}
    \label{f:ConnectedNetwork}
        \vspace{-2mm}
\end{figure}
The communication graph structure is the same as the physical connection graph (\cref{f:ConnectedNetwork}), with all the per unit line voltages chosen to be $V_i=V_j=1$ and line coefficients of the power flow are taken as $B_{12}=B_{21}=B_{13}=B_{31}=$ $B_{23}=B_{32}=B_{14}=B_{41}=0.545$ p.u. and $B_{24}=B_{42}=B_{34}=$ $B_{43}=0$ \cite{bevrani2014robust}. We assume the measurement model to be graph Laplacian:
\begin{align*}
    y_1 &= (\theta_1-\theta_2)+(\theta_1-\theta_3)+(\theta_1-\theta_4),\\
    y_2 &= (\theta_2-\theta_3)+(\theta_2-\theta_1),\\
    y_3 &= (\theta_3-\theta_1)+(\theta_3-\theta_2),\\
    y_4 &= (\theta_4-\theta_1)+(\theta_4-\theta_{ref}).
\end{align*}
The above measurement model implies that each individual area measures the sum of the phase difference between itself and physically connected areas through net tie-line active power flow measurement. We assume that we can measure the absolute phase angle of area 4 by comparing it with known reference $\theta_{ref}=0$. We bound the deviation in turbine power by choosing the performance variable state as: $z_t = [P_{t_1}~P_{t_2}~P_{t_3}~P_{t_4}]$ and want to obtain the same level of privacy in frequency deviation for each area in the system $\omega_i$.

\begin{figure}[ht!]
    \centering
    \includegraphics[width=1\linewidth]{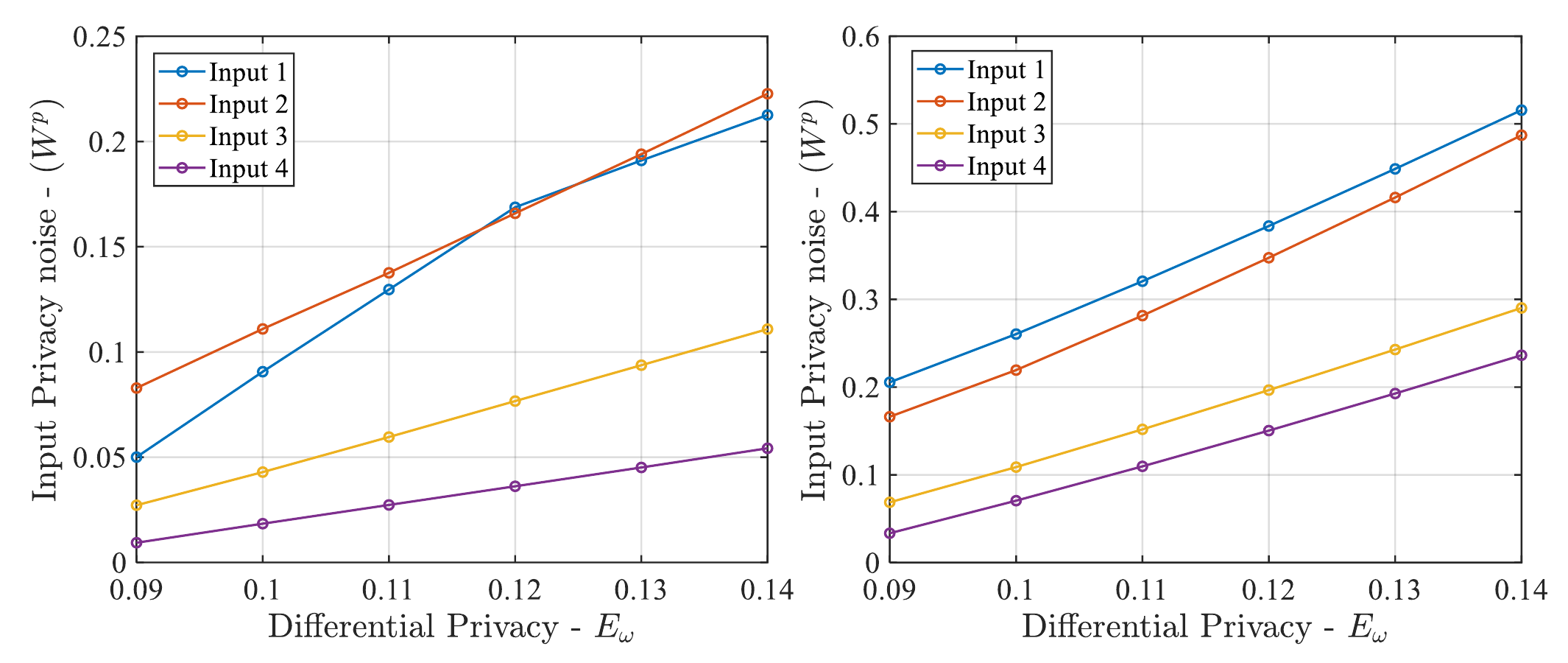}
    \caption{Optimal input private noise for given values of differential privacy for (L) an adversary with access to communication channels ($y^a_k = \bar{y}_k$), and (R) an adversary with direct access to measurements ($y^a_k = y_k$).}
    \label{fig:InputNoise}
\end{figure}
\vspace{-1mm}
The design problem is to find the optimal privacy noise in control input and output channels and simultaneously design the controller to bound the covariance of the deviation in the turbine power $Z_{P_{t_i}}$ while preserving the privacy of the frequency deviation for each area $E_{\omega_i}$. Figure~\ref{fig:InputNoise} shows the optimal input privacy noises for each of the control inputs to obtain the same level of privacy for each of $E_{\omega_i}$. Notice that the intensity of the noise required is different in each channel based on the open loop and finally the closed-loop dynamics of each area. Also, the intensity of the noise increases with the level of privacy but there is a change in respective ratios of the noise intensity between different channels showing a non-scaled parameterization of the privacy noise and thus the need for the co-design of the noises and the controller. Also, notice that the amount of privacy noise required for the second case with ($y^a_k = {y}_k$) is more than the first case ($y^a_k = \bar{y}_k$) as there is no contribution from the output privacy noise towards the privacy of the system.

Figure~\ref{fig:OutputNoise} shows similar plots for the optimal output privacy noise for two cases that are based on adversarial capabilities. Notice that the amount of privacy signal required to obtain the desired performance bound increases with an increase in desired differential privacy. Moreover, the amount of privacy noise added in the output channel is much higher than the input channels as the output channel directly affects the state estimates and indirectly affects the system performance after passing through the controller dynamics, but the input channel directly affects the system performance and indirectly affects state estimates after passing through the system dynamics. Notice that the amount of noise added in the second case with ($y^a_k = {y}_k$) is zero as it does not help increase the differential privacy but adversely affects the system performance. 

\vspace{-2mm}
\begin{figure}[hb!]
    \centering
    \includegraphics[width=1\linewidth]{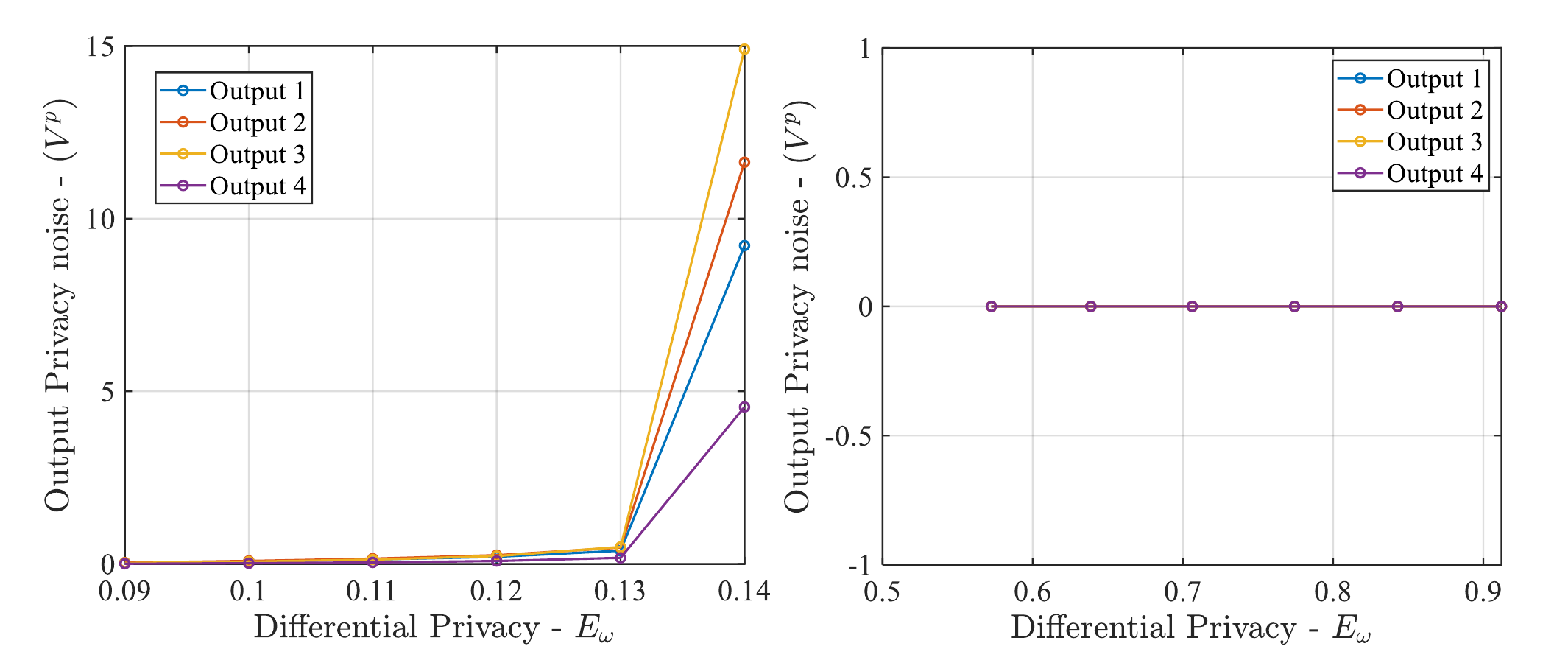}
    \caption{Optimal output private noise for given values of differential privacy for (L) an adversary with access to communication channels ($y^a_k = \bar{y}_k$), and (R) an adversary with direct access to measurements ($y^a_k = y_k$).}
    \label{fig:OutputNoise}
\vspace{-2mm}
\end{figure}
\vspace{-5mm}
\begin{figure}[ht!]
    \centering
    \includegraphics[width=1\linewidth]{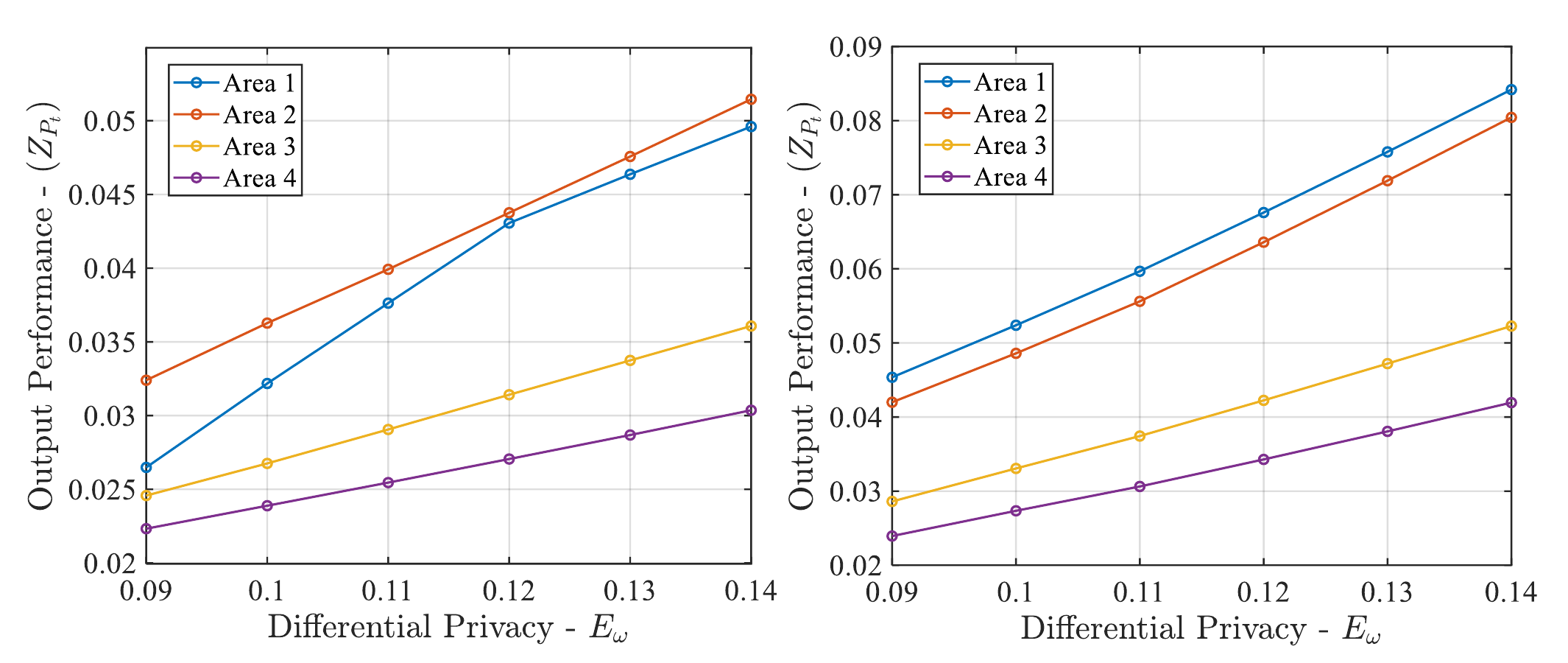}
    \caption{Optimal system performance norm for given values of differential privacy for (L) an adversary with access to communication channels ($y^a_k = \bar{y}_k$), and (R) an adversary with direct access to measurements ($y^a_k = y_k$).}
    \label{fig:Performance}
\vspace{-2mm}
\end{figure}
\vspace{-1mm}
Finally, \cref{fig:Performance} shows the plots for the system performance in terms of variance of deviation in turbine power for different areas. Notice that higher differential privacy results in a higher variance in deviation, i.e., an adverse effect on system performance. Notice that we minimized the performance loss for the given privacy level and thus different areas result in different performance levels. Also, the variance in deviation is higher for the second case with a stronger adversary with direct access to the measurement as only input privacy noise is effective in providing privacy which has a worse effect on system performance.

\vspace{-1mm}
\section{Conclusion}\label{s:conc}
\vspace{-1mm}
The paper showed that the joint design of differential privacy noise distribution and a general dynamic controller can be posed as a convex optimization problem using the Linear Matrix Inequalities framework. The framework adds privacy noise to both control input and system output to privatize the system’s state. The co-design problem also designs an optimal estimator from the perspective of the adversary with access to both communication channels and direct output measurements. 
The simulation results show the interplay between the controller gains and the privacy noise to obtain the desired level of privacy while minimizing the system performance as a measure of the variance of deviation from reference. The results show the effectiveness of input and output privacy noise based on the capabilities of the adversary and show the need for the co-design of the privacy noises with the controller.

\vspace{-1mm}
\bibliographystyle{IEEEtran}
\bibliography{Refs}

\end{document}